\documentclass[a4paper,aps,twoside,brazil,english,prd,amsmath,nofootinbib,
superscriptaddress,showpacs, reprint]{revtex4-1}

\usepackage[T1]{fontenc}
\usepackage[utf8]{inputenc}
\usepackage[brazil,english]{babel}
\usepackage{tensor}
\usepackage[usenames,dvipsnames]{xcolor}
 \definecolor{darkblue}{RGB}{0,0,150}
\usepackage{amssymb}
\usepackage{tensor}
\usepackage{graphicx}
\usepackage{nicefrac}
\usepackage{amsthm}
\usepackage{mathrsfs}

  \newtheorem{theorem}{Theorem}[section]

  \newtheorem{proposition}[theorem]{Proposition}

\usepackage[unicode]{hyperref}
\hypersetup{
  pdfinfo = {
    Author = {Alan Maciel, Morgan, Mimoso},
    Title = {(provisory)Static and Homogeneous fluids in dual null formalism}
  },
  colorlinks = true,
  breaklinks = true,
 linkcolor = blue,
 urlcolor = blue,
 citecolor = blue,
}
\usepackage{multirow}
\usepackage{enumerate}

\usepackage[%textsize=tiny,
            linecolor=orange,
            backgroundcolor=red!30,
            colorinlistoftodos]{todonotes}
\PassOptionsToPackage{normalem}{ulem}
\usepackage{ulem}

\newcommand{\Mov}[1]{{\color{black}{#1}}}

\newcommand{\Amv}[1]{{\color{black}{#1}}}

\newcommand{\ud}{\ensuremath{\mathrm{d}}}
\newcommand{\Lie}{\ensuremath{\mathcal{L}}}

\DeclareMathAlphabet{\mathpzc}{T1}{pzc}{m}{it}

\def\coxa{{\Huge%                              % aviso de revisão fácil de achar
C\kern-.1667em\lower.5ex\hbox{O}\-X\kern-.1667em\lower.5ex\hbox{A}\@}%
\index{CoXa}
}


\begin{document}
\begin{abstract}
The TOV equation appears as the relativistic counterpart of the classical condition for hydrostatic equilibrium. In the present work we aim at showing that a generalised TOV equation also characterises the equilibrium of models endowed with other symmetries besides spherical. We apply the dual null formalism to 
spacetimes with two dimensional spherical, planar and hyperbolic symmetries with a perfect fluid as the source. We also assume a Killing vector field orthogonal to the surfaces of symmetry, which gives us static solutions, in the timelike Killing field case, 
and homogeneous dynamical solutions in the case the Killing field is spacelike. 
In order to treat equally all the aforementioned cases, we discuss the definition of a quasi-local energy for the spacetimes with planar and hyperbolic foliations, since the Hawking-Hayward definition only applies to compact foliations. After this procedure, we are able to translate our geometrical formalism to the fluid dynamics language in a unified way, to find the generalized TOV equation, for the three cases when the solution is static, and to obtain the evolution equation, for the homogeneous spacetime cases.
Remarkably, we show that the static solutions which are not spherically symmetric  violate the weak energy condition (WEC). \Amv{We have also shown that the counterpart of the TOV equation for the spatially homogeneous models is just the familiar equation $\rho+P=0$\Mov{, defining a cosmological constant-type behaviour}, both in  the hyperbolic and spherical cases. This implies a violation of the strong energy condition in both cases, \Mov{added to the above mentioned violation }%in addition to that 
of the weak energy condition in the hyperbolic case}.
We  illustrate our unified treatment obtaining analogs of Schwarzschild interior solution,  for an incompressible fluid $\rho = \rho_0$ constant. \end{abstract}

%\title{()Static and homogeneous fluids in dual null formalism}

\title{% Dual null dynamics of relativistic  fluids with  symmetries
New perspectives on the TOV equilibrium from a dual null approach
}

\author{Alan Maciel}
\email{alan.silva@ufabc.edu.br}

\affiliation{Centro de Matemática, Computação e Cognição, Universidade Federal do ABC,\\
 Avenida dos Estados 5001, CEP 09210-580, Santo André, São Paulo, Brazil. }

\author{Morgan Le Delliou}
\email{delliou@lzu.edu.cn,delliou@ift.unesp.br}
\affiliation{Institute of Theoretical Physics, School of Physical Science and Technology, Lanzhou University, 
No.222, South Tianshui Road, Lanzhou, Gansu 730000, P R China}
\altaffiliation[Also at ]{Instituto de Astrof\'isica e Ci\^encias do Espa\c co, Universidade de Lisboa, Faculdade de Ci\^encias, 
Ed. C8, Campo Grande, 1769-016 Lisboa, Portugal}
%  \email{Second.Author@institution.edu}
% \affiliation{%
%  Authors' institution and/or address\\
%  This line break forced with \textbackslash\textbackslash
% }%

%\collaboration{MUSO Collaboration}%\noaffiliation

\author{Jos\'e P. Mimoso}%

\email{jpmimoso@fc.ul.pt}

\affiliation{
 Departamento de F\'{i}sica and Instituto de  Astrof\'{i}sica e Ci\^encias do Espa\c co, Faculdade de Ci\^{e}ncias da
Universidade de Lisboa, Campo Grande, Ed. C8 1749-016 Lisboa,
Portugal}

\maketitle

\section{Introduction}

When one describes spherical stars in equilibrium it is well known that %in which 
the matter distribution must satisfy the Tolman-Oppenheimer-Volkoff (TOV) equation \cite{Tolman:1934za,oppenheimer-1939a}.
%however suffers the exception, revealing a deeper insight is at play, is given by 
The TOV equation appear as the relativistic counterpart of the classical condition for hydrostatic equilibrium. It gives a first approximation to describe  virtually any body in the sky which is large enough such that its dynamics is dominated by gravity, and stationary enough to enable us to assume it to be in static  equilibrium, such as planets and stars. Due to their evident relevance, many solutions for this type of configurations have been found \cite{Zurek:1984zz, Delgaty:1998uy, Berger:1985vx,Fodor:2000gu,Dev:2000gt, Rahman:2001hp,Gorini:2008zj,Zou:2013bra,Herrera:2013fja,Ozel:2016oaf,Herrera:2017hua,Faraoni:2018xwo,Martins:2018ihu}, and more specifically, different formalisms \cite{Arbanil:2013pua,Carloni:2017rpu, Carloni:2017bck, Isayev:2018hqx,Papadopoulos:1999kt} and solutions generating techniques have been developed \cite{lake-2003,Boonserm:2006up}. Extensions of the TOV equation have also been investigated in the framework of modified gravity theories~\cite{Olmo:2012er,Glampedakis:2015sua,Jain:2015zcs,Velten:2016bdk,Wojnar:2016bzk,Bronnikov:2016odv,Astashenok:2017dpo}. Yet a unified characterization of the underlying features of the TOV equation  has attracted little attention, and this is what concerns us in the present work. 
{
As it happens with the generalised perception of  Birkhoff theorem \cite{Birkhoff1923} as being restricted to the spherically symmetric case, which is a misled assumption as shown for instance  by C. Bona \cite{Bon1988},and more recently {discussed} %\sout{further extended} 
in \cite{Maciel:2018tnc}, a similar idea is very much spread regarding the TOV equation. In the present work we aim at showing that this is a restrictive view, and that a generalised TOV equation also characterises the equilibrium of models endowed with other symmetries besides the spherical.  
}

Focusing on GR fluid dynamics 
%\cite[for a review of the pioneering work on this viewpoint see for instance][and references therein]{Ehlers:1993gf,Ellis:1971pg,ellis-1999}
(for a review of the pioneering work on this viewpoint see for instance, and references therein \cite{Ehlers:1993gf,Ellis:1971pg,ellis-1999}), we
% The interpretation of GR as general relativistic fluid dynamics is a most convenient standpoint to tackle any general relativistic problem, regardless of the scale of its domain.  This perspective 
can trace it back to Hawking and Penrose's singularity theorems \cite{hawking}, and it eventually %this 
enables one to tackle ``small scale'' problems with the same tools are those applicable to ``large scale'' ones. In this context, an approach using the properties of light cones is most likely to reveal the structures of spacetime both at small and large scales. 
% [JPM: The work of J. Ehlers, and of Kundt and Schicking is likely to predate Hawking and Penrose's...]

The dual null formalism\footnote{{We follow here the nomenclature coined by Sean A. Hayward \cite{Hayward:1993dq, Hayward_1993, Hayward_1993b} although %\Mov{\sout{the}} 
in some references this formalism is referred to as double-null.}}
offers  a description of the spacetime based on the properties of the optical flow. The latter is characterised by two linearly independent null congruences  which are orthogonal to some codimension-two foliation of the spacetime.  This approach  shares some of the convenience of the choice of dual null coordinates, but it has the significant advantage of being a coordinate free formalism.  Furthermore it also reveals, by construction, the causal structure of the spacetime in a natural way. It has been originally introduced to study general relativistic problems associated with the behaviour of dynamical black holes \cite{Hayward:1993mw, Senovilla:2011fk},  but it is also most convenient to %apply %also to the
analyse other diverse questions. For instance, it has been considered in connection with the definition of energy in more general geometries \cite{Hayward:1993ph}, with the gravitational collapse of fluids \cite{Maciel:2015vva}, and even with the definition of generalized horizons in modified gravity \cite{Maciel:2015ypv}. The dual null formalism has %also 
been useful  to explicit the "linear" behavior of gravity for sources that satisfy the hypotheses of the Birkhoff theorem \cite{Maciel:2018tnc}, as well. 

Here, we apply the dual null formalism to analyse, in a unified way, the spacetimes which admit a codimension-two foliation with constant curvature leaves. This comprises the spherical, planar and hyperbolic %al 
symmetries, sourced by a perfect fluid. Aside from the Killing vector fields that are tangent to those surfaces of symmetry, we assume the existence of an additional symmetry generated by a  Killing vector field %which is 
orthogonal to those surfaces at each event. A particular case of this setup, where the symmetry is spherical and the Killing vector is timelike, corresponds to the %case of a 
spherically symmetric perfect fluid in hydrostatic equilibrium, which leads us to the well-known TOV equation. As we will show this celebrated equation arises most naturally in the dual-null framework which, moreover, allows us to  generalize it for the planar and hyperbolic cases. 
\Mov{To the best of our knowledge, this generalisation of the geometry underlying
%of application for 
the TOV equation, stepping %going 
beyond spherical symmetry, has never been seen before, and leads to consequences which are far from trivial. 
%Furthermore,

Since in planar and hyperbolic geometries the  spatial hypersurfaces are open, this extension requires the novel introduction of a mass-energy "parameter". In fact one needs to promote a generalization of the Misner-Sharp/Hawking-Hayward definitions of the mass-energy distribution, which overcomes the %induced 
problem of the divergence of the latter quantities due to the natural threading with infinite surfaces. 
 %[[perhaps this sentence is better placed in the conclusion?]]

Finally and unexpectedly, the third novelty of our explorations of the planar and hyperbolic geometries stems from  the physically significant appearance of violations of the weak energy conditions in order to maintain \Amv{hydrostatic} equilibrium. 
This takes the form of negative mass, physically translating 
%Those violations are indications of
repulsive curvature effects,} \Amv{which suggest a link to repulsive source models, as those proposed to mimic dark energy, generate bouncing universes, or support classical wormholes \cite{Cattoen2005, Lobo:2009du}.
}

When the metric is characterised by a spacelike Killing vector% which is spacelike
, we have spatially homogeneous spacetimes that, as we will show, %will thus 
correspond to some of the Bianchi spacetimes, or Kantowski-Sachs, as expected. \Amv{The hydrostatic equilibrium on those spacetimes are only possible when their source is a cosmological constant in the non flat cases, also implying the violation of energy conditions.}

%{%This 
%The consideration of \Amv{infinite mass-energy distributions} in those spacetimes will lead us to put forward a generalization of the Misner-Sharp/Hawking-Hayward definitions.} \Amv{As a consequence, we establish a relationship between the hydrostatic equilibrium, the symmetry type and the energy conditions. In particular, the existence of static solutions in planar and hyperbolic symmetries is subjected to a violation of the weak energy condition, which suggest a link to repulsive source models as those proposed to mimic dark energy, generate bouncing universes or support classical wormholes \cite{Cattoen:2005dx}}

We proceed as follows: In Sec.~\ref{sec:Defs} we start by giving a short introduction to the dual null formalism, using from the onset the symmetries assumed in our class of problems %in order 
to simplify %our 
expressions. In Sec.~\ref{sec:OrthoKillingV}%ing two-expansion,
We then %Then we 
prove a proposition that states that the Killing vector two-expansion %has 
always vanishes. In order to interpret the geometrical quantities that appear, and to  establish their  underlying physical content, we discuss, in Sec.~\ref{sec:Energy}, the mass-energy definition in these spacetimes. From the property found in  Sec.~\ref{sec:OrthoKillingV}, %this property 
we %will 
show that one derives either the equation of hydrostatic equilibrium when the Killing is timelike, in Sec.~\ref{sec:subTKV}, or %else, 
the evolution equation when the Killing is spacelike, in Sec.~\ref{sec:subSKV}%
%that require foliations in terms of compact leaves to planar and hyperbolical foliations where the leaves are non compact, but are homogeneous
. Finally, in Sec.~\ref{sec:PHS} we look for planar and hyperbolic symmetric analogs of Schwarzschild interior solution, by assuming that the fluid is uncompressible and solving the unified 
TOV equations.

\section{Main Assumptions and Definitions}\label{sec:Defs}

We consider metrics that have a codimension-two maximally symmetric foliation, 
and %which 
can be written as
\begin{gather}
\ud s^2 =  N_{ab} \ud x^a \ud x^b + Y^2(x^c) \left( 
\ud \theta^2 + S_{\epsilon}^2  \ud \phi^2 \right)\,, 
\label{eq:metric-dualnull}
\end{gather} 
where
\begin{gather}
 S_\epsilon = \left\lbrace \begin{tabular}{r}
                            $\sin \theta$, for $\epsilon = 1$\\
                            $1$, for $\epsilon = 0$\\
                            $\sinh \theta$, for $\epsilon = -1$\\
                           \end{tabular}~~,\right.
\nonumber
\end{gather}
and where we divide the tangent space $\mathcal{T}$ at each event in two orthogonal subspaces $\mathcal{T}=\mathcal{N} \oplus \mathcal{S}$. Here %, where 
$\mathcal{S}$ is the subspace generated by the orbits of $(\theta, \phi)$ and $\mathcal{N}$, %is 
the subspace of $\mathcal{T}$ %which is 
orthogonal to $\mathcal{S}$. The $x^a$ coordinates are chosen %to be 
orthogonal to $\mathcal{S}$, which gives the metric in the warped sum form of Eq.~\eqref{eq:metric-dualnull}.

We denote $s_{ab}=Y^2 \gamma_{ab}$ the induced metric in each leaf of the foliation where $Y(x^c)$ { is the warp factor. }
Evidently, $\gamma_{ab}:=\delta^\theta_a\delta^\theta_b+S_{\epsilon}^2 \delta^\phi_a\delta^\phi_b$ has constant curvature and does not depend on the 
coordinates $x^a$ which identify each leaf $\Sigma_{x^c}$, defined as the locus spanned by the orbits of $\theta$ and $\phi$ for fixed $x^c$.
We define an orthonormal two dimensional basis $(n^a\,,e^a)$ for $\mathcal{N}$, whose induced metric is $N_{ab}$, according to Eq.~\eqref{eq:metric-dualnull}. This basis satisfies
\begin{gather}
- n^a n_a = e^ae_a = 1\,, \quad n^a e_a = n^a 
s_{ab} = e^a s_{ab} = 0\,.
\end{gather}

We may also define a dual null basis for the same subspace from $n^a$ and $e^a$ by
\begin{gather}
 k^a = \frac{1}{2} \left (n^a + e^a \right) \,, \quad l^a = \frac{1}{2} 
\left(n^a - e^a \right)\,, \nonumber\\
 n^a = k^a + l^a \,, \quad e^a = k^a - l^a \, ,\label{eq:NullNonNull}
\end{gather}
which satisfies
\begin{gather}
k^ak_a = l^a l_a =0 \,, \quad  k^a l_a = -\frac{1}{2}\,.
\end{gather}
{The metric $g_{ab}$ can be written as
\begin{gather}
    g_{ab}= \frac{2}{k^c l_c}k_{(a}l_{b)}+s_{ab}\,. \label{eq:gdecomposition}
\end{gather}}
We associate the null expansion for each null vector %, defined 
as follows
\begin{gather} 
\Theta_{k} =\frac{1}{2}s^{ab} \Lie_k s_{ab}= \frac{1}{2} Y^{-2}\gamma^{ab}\Lie_k Y^2 \gamma_{ab} = \frac{2}{Y} k^a \partial_a Y\,.\label{eq:2expansiondef}
\end{gather}
We may extend the definition of null expansion to %for 
timelike and spacelike vectors in $\mathcal{N}$, calling %and we call 
it the two-expansion, since it measures %they measure 
the rate of variation of area, as in the null case. We may define the mean curvature form $\mathcal{K}_a = \partial_a \ln Y^2 $, such that, we obtain for the two-expansion $\Theta_{(u)}$ of any vector $u^a$ in $\mathcal{N}$
\begin{gather}
    \Theta_{(u)} = u^a \mathcal{K}_a\,.\label{eq:MeanCurvForm}
\end{gather}

We describe our spacetimes  % with by 
by means of the behaviour of the null expansion,  casting the Einstein equations, $G_{ab} = 8 \pi T_{ab}$ in terms of expansions, i.e., by writing the Raychaudhuri equations \cite{Hayward:1993mw,wald, hawking, carroll-2004}. In the latter equations $T_{ab}$ is, as usual, the energy-momentum tensor $T_{ab}$, and  we express it as a general fluid under a 1+1+2 decomposition along $n^a$ and $e^a$ \cite{Clarkson:2002jz} thus reading
\begin{align}
    T_{ab} = & \left(\rho + P\right)n_a n_b + Pg_{ab} +2q_{(a} n_{b)} +\Pi_{ab}\; . \label{EMT_general}
\end{align}
In the latter decomposition $\rho$ is the  energy density, and $P$ is the isotropic pressure, both measured by an observer moving with 4-velocity $n^a$, $q^a=q e^a + \frak{q}^a$ is the heat flow vector, decomposed into its scalar part along $e^a$ and a 2-vector $\frak{q}^a$ on the maximally symmetric surfaces, and, finally, $\Pi_{ab}$ is the anisotropic stress tensor.  $\Pi_{ab}$ is also decomposed  as $\Pi_{ab}=\Pi P_{ab} + 2 \Pi_{(a}e_{b)} + \Pi_{\left\langle ab \right\rangle_s}$ into a scalar part along the flow-orthogonal  symmetric tensor $P_{ab}= s_{ab}-2 e_a e_b$, which is traceless, into a vector part semi-orthogonal to the maximally symmetric surface with inner vector $\Pi^a$, and into a fully embedded tensor $\Pi_{\left\langle ab \right\rangle_s}$. 

In terms of the metric (\ref{eq:gdecomposition}), the energy-momentum tensor can be expressed as
\begin{align}
    T_{ab} = & \left[\rho + P + 2\left(q - \Pi\right)\right]k_a k_b %\\
             +\left[\rho + P - 2\left(q + \Pi\right)\right]l_a l_b\nonumber\\
            & +2 k_{(a} l_{b)}\left[\rho - P + 2\Pi\right] + \left[ P - \Pi\right]s_{ab} + 2 \frak{q}_{(a} \left[k+l\right]_{b)}\nonumber\\
            & + 2 \Pi_{(a}\left[k-l\right]_{b)} + \Pi_{\left\langle ab \right\rangle_s}\; ,
\end{align}
and thus  its projected components yield
\begin{align}
    T_{ab}k^a k^b = & \frac{1}{4}\left[\rho + P - 2\left(q + \Pi\right)\right]\,,\\
    T_{ab}k^a l^b = & \frac{1}{4}\left[\rho - P + 2\Pi\right]\,,\\
    T_{ab}l^a l^b = & \frac{1}{4}\left[\rho + P + 2\left(q - \Pi\right)\right]\,.
\end{align}
The Raychaudhuri and constraint equations then %. They 
read
\begin{subequations}
\begin{gather}
\Lie_k \Theta_{(k)} = \nu_k \Theta_{(k)} - \frac{\Theta_{(k)}^2}{2}  - 
8\pi T_{ab}k^ak^b\,, \label{eq:uu-matter}\\
\Lie_l \Theta_{(l)} = \nu_l \Theta_{(l)} -  
\frac{\Theta_{(l)}^2}{2}  - 8\pi T_{ab}l^a l^b \,, \label{eq:vv-matter}\\
 \Lie_k \Theta_{(l)} +\Lie_l \Theta_{(k)}= -  \Theta_{(l)} 
\nu_k - \Theta_{(k)} \nu_l   - \nonumber\\ 2\,\Theta_{(k)} \, \Theta_{(l)} +
\epsilon\frac{ 2\,k^a l_a}{Y^2} + 
16\pi\,T_{ab}k^a l^b \label{eq:uv-matter}\,,
\end{gather}
\end{subequations}
where we included the inaffinities $\nu_k$ and $\nu_l$, defined as
\begin{gather}
    \nu_k = \frac{1}{k^c l_c}l^b k^a \nabla_a k_b \, \qquad \nu_l = \frac{1}{k^c l_c}k^b l^a \nabla_a l_b\,.
\end{gather}

In this work we %intend to 
adapt our vector basis to a fluid source, such that $n^a$ gives its flow. Therefore, it will be useful to relate our quantities to this flow. By construction, the flow $n^a$ %in
is %our problem will 
always %be 
orthogonal to the surfaces of symmetry and will be characterized by two quantities
\begin{gather}
   \mathcal{A} = e^a \dot{n}_a = e^an^b \nabla_b n_a\,, \qquad 
   \mathcal{B} = e^a n'_a = e^a e^b \nabla_b n_a \, .
\end{gather}
The scalar $\mathcal{A}$ gives us the acceleration of the flow, a positive sign meaning that the acceleration is outwards in the spherical, compact case. The scalar $\mathcal{B}$ gives the change of direction of %the 
$n^a$ as we travel along $e^a$. It is the $e-e$ component of the extrinsic curvature $K_{ab}$ of the 3-space orthogonal to this flow, since 
\begin{gather}
    K_{ab}= \frac{1}{2}\Lie_n h_{ab}\,,
\end{gather}
where $h_{ab}= g_{ab} + n_a n_b$. We may also write
\begin{gather}
    h_{ab} = e_a e_b + Y^2 \gamma_{ab}\,,
\end{gather}
which gives
\begin{gather}
    K_{ab}= \mathcal{B} \, e_a e_b + \frac{\Theta_{(n)}}{2} \, Y^2 \gamma_{ab}\,.\label{eq:extrinsic}
\end{gather}

The trace of Eq.~\eqref{eq:extrinsic} gives us the flow of the volumetric expansion $\Theta_3 = \nabla_a n^a = K_a^a$ as
\begin{gather}
    \Theta_3 = \mathcal{B} + \Theta_{(n)}
\end{gather}

In order to relate %ing 
our quantities with the flow scalars, we compute the shear scalar $\sigma$, by taking the traceless part of $K_{ab}$. We obtain
\begin{gather}
    \sigma = \frac{\Theta_{(n)}}{6} - \frac{\mathcal{B}}{3}\,, \label{eq:B}
\end{gather}
which implies
\begin{gather}
    \frac{\Theta_3}{3}+ \sigma =\frac{\Theta_{(n)}}{2}\,,\label{eq:theta3}
\end{gather}
in agreement with the result obtained in Ref.~\cite{Maciel:2015vva}. 

Using the inaffinities of the null basis vectors, %and 
$\mathcal{A}$ and $\mathcal{B}$ can be expressed as
\begin{gather}
    \mathcal{A} = \nu_k - \nu_l\,, \qquad \mathcal{B} = \nu_k + \nu_l\,.
\end{gather}

\section{Orthogonal Killing Vector}\label{sec:OrthoKillingV}

We now %Now, we 
assume that our metric has a Killing vector orthogonal to maximally symmetric surfaces{. { Our symmetry requirements imply that it commutes with the symmetry generators on the foliation.} We denote this }%, which means the 
%existence of 
hypersurface orthogonal Killing vector field $\chi^a$. It satisfies the Killing equation,

\begin{gather}
 \Lie_\chi g_{ab} = 0\,.
\end{gather}

\begin{proposition}
 If a spacetime is 
 described by a metric of the form \eqref{eq:metric-dualnull} and admits an orthogonal Killing vector $\chi^a \in \mathcal{N}$, then  $\Theta_{\chi} = 0$.
% \begin{itemize}
%  \item $\Theta_{\chi} = 0$.
% \end{itemize}
\label{prop:killing}
\end{proposition}

\begin{proof}
 We may write, from Eq.~\eqref{eq:2expansiondef}, $\Theta_{\chi} =\frac{1}{2}s^{ab}\Lie_\chi s_{ab}$ \\$= \frac{1}{2} Y^{-2}\gamma^{ab} \Lie_{\chi} Y^2 \gamma_{ab}$, 
and
 \begin{gather}
  g_{ab} = N_{ab} + Y^2 \gamma_{ab}\,.
 \end{gather}
Then
\begin{gather}
 0 = Y^{-2}\gamma^{ab} \Lie_{\chi} g_{ab} = Y^{-2}\gamma^{ab}\Lie_{\chi} N_{ab} + 
2\Theta_{\chi}  \nonumber\\
 = -Y^{-2}N^{ab}\Lie_{\chi} \gamma_{ab} + 2 \Theta_{\chi}\,. \label{eq:thetachi}
\end{gather}
However
\begin{gather}
 \Lie_{\chi}\gamma_{ab} = 0\,,
\end{gather}
since $\chi^a$ does not {admit components }%depend on the 
%coordinates 
in $\mathcal{S}$
{and }%nor 
${\gamma}_{ab}$ {doesn't }depend on  coordinates along $\mathcal{N}$. 
Therefore, Eq.~\eqref{eq:thetachi} implies that $\Theta_\chi = 0$.
\end{proof}

{Consequently}%Therefore
, if there is an extra symmetry with orbits orthogonal to those of the maximally symmetric leaves of the foliation, the two-expansion
of its generator vanishes. This also implies that if $\ud Y$ is spacelike, then $\chi_a$ is timelike and vice-versa. 
If $\ud Y$ is null, the Killing vector will also be null.

\section{Mass-Energy}\label{sec:Energy}

In order to properly interpret our spacetimes, we have to understand their mass-energy content. There is 
a 
widely known mass-energy definition suitable to the spherically symmetric case, namely the Misner-Sharp mass-energy \cite{Misner:1964je,Hayward:1994bu},  
defined regardless of asymptotic assumptions. However, as we also intend to analyze nonspherical spacetimes in this work, we are 
lead to a more general mass-energy definition such as the Hawking-Hayward's one (hereafter HH) \cite{Hayward:1993ph,Faraoni:2015cnd}. The HH mass-energy gives the mass-energy content inside a closed compact surface in terms of an integral over that surface, in a manner similar to the Gauss law in Newtonian gravity. This quasilocal mass-energy has been explored in different contexts, such as seen in Refs.~\cite{Hayward:1997jp,Prain:2015tda,Faraoni:2019rdj} .

In the case where the $\Sigma_{x^c}$ are spheres, which are compact, we can compute the Hawking-Hayward mass-energy  enclosed by $\Sigma$ (we drop the $x^c$ index for short) as
\begin{gather}
    M_\Sigma = \frac{1}{8\pi}\sqrt{\frac{A}{16\pi}}\int_\Sigma \left[\mathcal{R} -\right. \nonumber\\
    \left.\frac{1}{k^a l_a}\left(\Theta_{(k)} \Theta_{(l)} -\frac{1}{2}\sigma_{(k)\,ab}\sigma^{(l)\,ab} - 2\omega_a \omega^a \right)\right] \ud \Sigma \label{eq:HHmass}
\end{gather}
where $\mathcal{R}$ is the two-dimensional Ricci scalar and $A$ is the area of $\Sigma$,  $\sigma_{(k)\,ab}$ and $\sigma^{(l)\,ab}$ are the two-dimensional shear tensors along $\Sigma$, associated with the $k$ and $l$ congruences, respectively, and $\omega^a$ is the twist vector given by the projection on $\Sigma$ of the commutator of the null basis vectors. We have included a factor of $-\dfrac{1}{k^a l_a}$ in the optical scalars part of the mass-energy, compared with 
the formula present in Ref.~\cite{Hayward:1993ph}, in order to take account of our different normalization of the null normals.

Our symmetry assumptions imply that the only non-vanishing optical scalar 
on $\Sigma$ is the null expansion. Therefore, the Hawking-Hayward mass-energy is reduced to
\begin{gather}
    M_\Sigma = \frac{1}{8\pi}\sqrt{\frac{A}{16\pi}}\int_\Sigma \left[ \mathcal{R} - \frac{1}{k^a l_a}\Theta_{(k)} \Theta_{(l)} \right]\, \ud \Sigma \label{eq:HHmassreduced}
\end{gather}
Since we assume that $\Sigma$ is maximally symmetric, we have $\mathcal{R}=\dfrac{2 \epsilon}{Y^2}$. We also have
\begin{gather}
    \Theta_{(k)} \Theta_{(l)} = k^a \partial_a \ln Y^2 \, l^b \partial_b \ln Y^2 = k^a l^b \partial_a \ln Y^2 \partial_b \ln Y^2 =\nonumber\\
    = k^{(a}l^{b)}\partial_a \ln Y^2 \partial_b \ln Y^2 = \frac{k^cl_c}{2} g^{ab}\partial_a \ln Y^2 \partial_b \ln Y^2 = \nonumber\\
    \nonumber\\ \frac{1}{2}k^c l_c || \ud \ln Y^2||^2 = k^c l_c \frac{2}{Y^2} || \ud Y ||^2 \Rightarrow \nonumber\\
    \frac{\Theta_{(k)} \Theta_{(l)}}{k^cl_c} = \frac{2}{Y^2}||\ud Y||^2\,,
\end{gather}
{where we used Eq.~\eqref{eq:gdecomposition} in the fourth step.}

For the spherical case $\epsilon = 1$ and $A = 4\pi Y^2$, we obtain the known interpretation of $|| \ud Y||$ in terms of the Misner-Sharp mass-energy, which coincides with the Hawking-Hayward one
\begin{gather}
M_\Sigma = \frac{Y}{2}\left(1 - || \ud Y||^2 \right) \Leftrightarrow
|| \ud Y ||^2  = 1 - \frac{2 M}{Y}\,.
\end{gather}

In the planar and hyperbolic cases ($\epsilon = 0$ and $\epsilon = -1$, respectively), the Hawking-Hayward mass is not  conveniently defined for the integration domain set by our preferred foliation, as it requires a closed compact surface.

 In this work, we aim to treat all three symmetry types in the same manner. Therefore, we need to find a mass-energy definition which might be equivalent to the HH mass-energy, but suitable to deal with (instead of adapted for) non compact domains in order to take advantage of the planar or hyperbolic symmetry. We can make such an extension of the HH mass-energy, as long as their boundary correspond to a pair %set 
$\Sigma$ of symmetric two-surfaces of symmetry corresponding to the same warp factor $Y$. Of course, those domains are infinite and have an infinite mass-energy content in general. However, as they are homogeneous along the surfaces of symmetry, we can successfully %use 
adapt the HH mass-energy definition in order to obtain a finite \emph{mass-energy parameter} with those cases. They then %, that 
describe an infinite mass-energy distribution, %but 
homogeneous along the surfaces of symmetry, with a finite density. 

We proceed by first making the replacement
\begin{gather}
   \dfrac{1}{8\pi}\sqrt{\dfrac{A}{16 \pi}}\to\dfrac{Y}{4\pi \kappa}, \label{eq:replacement}
\end{gather}
in order to %\Mov{\sout{just}} 
keep its dimensionality, and %\Mov{\sout{but}} 
eliminating the explicit dependence on the area of $\Sigma$. Evidently, by setting $\kappa = 4$ we recover the Hawking-Hayward mass-energy in the spherical case. This step is justified by the fact that originally this factor was introduced to correct the dimensionality of the mass-energy, and to make it match the Arnowitt-Deser-Misner (ADM) mass~\cite{Arnowitt:1959ah}, where both are well defined. Since our symmetric spacetimes allow an %\sout{"areal coordinate"}
%{\bf [proposals]}
"areal scalar" %/\sout{"areal scale"} 
as the warp factor Y, we can replace $\sqrt{A}$ by $Y$ as the quantity with dimension of length associated to each surface of symmetry.

We then 
define the quasi-local mass-energy parameter $\mu(Y)$ by
\begin{gather}
M_\Sigma = \frac{\mu(Y)}{4\pi} \int S_\epsilon^2 (\theta) \ud \theta \ud \phi\,,\label{eq:mudefinition}
\end{gather}
and we write
\begin{multline} 
    \frac{Y}{4\pi \kappa} \left[\mathcal{R} - \frac{\Theta_{(k)} \Theta_{(l)}}{k^c l_c}\right]\int_\Sigma \ud \Sigma = \\
  \frac{Y}{4 \pi \kappa} \left[2\epsilon - 2||\ud Y||^2\right]\int S^2_\epsilon (\Theta)\ud \theta \ud \phi\,.\label{eq:HHmassepsilon}
\end{multline}
We equate Eqs. \eqref{eq:mudefinition} and \eqref{eq:HHmassepsilon} and eliminate the improper area integral on both sides
\begin{gather}
|| \ud Y ||^2 = \epsilon - \frac{\kappa \mu (Y)}{2Y} \,. \label{eq:muequation}
\end{gather}

%\sout{We can also obtain Eq.~\eqref{eq:muequation} by a an alternative route,}
An alternate route to Eq.~\eqref{eq:muequation} can be obtained by computing the HH mass-energy in a finite domain, %\Mov{\sout{that is}} 
symmetric with respect to the central plane or wire,  %\Mov{\sout{that correspond to}} 
$Y = 0$, %  in those spacetimes 
and taking the limit where the domain tends to be the whole surface. The finite integration domain  consist of the union of 
\begin{enumerate}
    \item a subset of the $\Sigma_Y$, that we denote $\Gamma_r$, bounded by a circle $\gamma\Mov{_r}$ of radius $r$ on the $(\theta, \phi)$ coordinate plane and 
    \item a compact surfaces given by the surfaces $\Delta_r$ defined by $\gamma_r$ transported along $Y$ orbits.
\end{enumerate} 
It forms
%compact surfaces given by \sout{the patch of}\Mov{\bf[we never defined a $\Delta_r$ so it is defined as the transport of $\gamma_r$]} the surfaces $\Delta_r$ \sout{given}\Mov{defined} by $\gamma_r$ transported along $Y$ orbits, forming 
a closed surface, corresponding to a part of a cylinder bounded by $Y = \text{constant}$ surfaces in the space of coordinates $(Y, \theta, \phi)$. Therefore, the HH mass-energy enclosed by those surfaces will by finite, and given by
\begin{equation}
    M_r = \frac{1}{8\pi}\sqrt{\frac{A_r}{16\pi}}\left(\int_{\Gamma_r} (\dots) S_\epsilon^ 2\ud\theta \ud \phi + \int_{\Delta_r} (\dots) \ud\Delta \right) \label{eq:MR}
\end{equation}
where $(...)$ replaces the integrand of Eq.~\eqref{eq:HHmass}. In the limit $r \to \infty$, the first integral in Eq.~\eqref{eq:MR} scales as $r^ 2$ while the second one scales as $r$. This means that, in the limit $r \to \infty$, and repeating the replacement in Eq.~\eqref{eq:replacement} we obtain
\begin{gather}
    \frac{M_r}{A_r} \to \frac{\mu(Y)}{4\pi Y^ 2}
\end{gather}

Equation~\eqref{eq:muequation} coincides with the known mass function %\cite[see][Eq. 15.7a]{stephani-exact} 
[see Eq. (15.7a) in \cite{stephani-exact}]
which appears as we integrate the Einstein equations of specific spacetimes with metrics of the form \eqref{eq:metric-dualnull} for planar and hyperbolic symmetries. From now on, we will consider Eq.~\eqref{eq:muequation} with the choice $\kappa = 4$ as the mass-energy definition.

% \section{Equilibrium equations}

% Here we assume the presence of a perfect fluid source, of timelike flow $n^a$, 
% whose energy-momentum tensor is
% \begin{gather}
% T_{ab} = (\rho + P) n_a n_b + g_{ab} P\,,
% \end{gather}
% with $n^a n_a = -1$.

% The Einstein equations, $G_{ab} = T_{ab}$ in terms of expansions read
% \begin{subequations}
% \begin{gather}
% \Lie_k \Theta_{(k)} - \nu_k \Theta_{(k)}+ \frac{\Theta_{(k)}^2}{2}  + 
% T_{ab}k^ak^b=0\,, \label{eq:uu-matter}\\
% \Lie_l \Theta_{(l)} - \nu_l \Theta_{(l)}  + 
% \frac{\Theta_{(l)}^2}{2}  + T_{ab}l^a l^b = 0\,, \label{eq:vv-matter}\\
% \frac{1}{2}\left( \Lie_k \Theta_{(l)} +\Lie_l \Theta_{(k)} +  \Theta_{(l)} 
% \nu_k + \Theta_{(k)} \nu_l  \right) + \nonumber\\ \Theta_{(k)} \, \Theta_{(l)} - 
% \frac{\epsilon k^a l_a}{r^2} = 
% T_{ab}k^a l^b \label{eq:uv-matter}\,,
% \end{gather}
% \end{subequations}

\section{Evolution equations}\label{sec:evolution}

\subsection{Timelike Killing Vector}\label{sec:subTKV}

We assume $\chi^a \chi_a < 0$. In this case, the spacetime is static, and $n^a \sim 
 \chi^a$. 
Therefore, from Proposition \ref{prop:killing}, $\Theta_{(n)} =0$ everywhere, and $\ud Y$ is spacelike, since it is orthogonal to $n^a$. If $\Theta_{(n)}$ 
vanishes everywhere, this means that the fluid has no radial velocity, 
therefore we are dealing with a static fluid with a flow parallel to the Killing vector field.

In order to characterize its static equilibrium,  we need to compute the 
derivative of the flow 2-expansion along the flow itself:
\begin{gather}
\Lie_n \Theta_{(n)} = 0\,,
\end{gather}
since $\Theta_{(n)} = 0$ everywhere.

We may write $\Lie_n \Theta_{(n)}$ in terms of the null expansions as
\begin{gather}
    \Lie_n \Theta_{(n)} = \Lie_k \Theta_{(k)} + \Lie_l \Theta_{(l)} + \Lie_l \Theta_{(k)} + \Lie_k \Theta_{(l)}\,.
\end{gather}

Substituting the Eqs.~\eqref{eq:uu-matter}, \eqref{eq:vv-matter}, and \eqref{eq:uv-matter}, we obtain
\begin{gather}
    \Lie_n \Theta_{(n)} = - \frac{(\Theta_{(k)}+  \Theta_{(l)})^2}{2} - \nonumber\\\Theta_{(k)}\Theta_{(l)} - \frac{\epsilon}{Y^2}- 8\pi T_{ab}e^ae^b + \mathcal{A} (\Theta_{(k)}-\Theta_{(l)})\,.\label{eq:Lienthetan}
\end{gather}

{Recall that }%Since 
we are assuming $\Theta_{(n)} = \Theta_{(k)} + \Theta_{(l)}=  0${, and that }% and 
$\Theta_{(k)} - \Theta_{(l)} = \Theta_{(e)}${, using Eqs.~(\ref{eq:MeanCurvForm}) and (\ref{eq:NullNonNull})}. We identify here 
$\Theta_{(k)}\Theta_{(l)}$ as the mass term, since it equals $ \dfrac{2}{Y^2} || \ud Y ||^2 = \dfrac{2}{Y^2}\left({ \epsilon} - \dfrac{2\mu(Y)}{Y}\right)$.

% We assume 
Taking the source to be a perfect fluid, then the energy momentum tensor (\ref{EMT_general}) reduces to
\begin{gather}
    T_{ab} = \rho n_a n_b + P (e_a e_b + s_{ab})\,.
\end{gather}
Contracting the conservation of the energy-momentum tensor with $e_b$
%{ \cite[Euler equation in][]{Mimoso:2009wj}}:
(Euler equation in \cite{Mimoso:2009wj}) we get
\begin{gather}
e_b \nabla_a T^{ab} = (\rho + P) \dot{n}^b e_b + e^a \nabla_a P = 0 \Rightarrow 
\nonumber\\
\mathcal{A} = - \frac{e^a \partial_a P}{\rho + P}\,.\label{eq:Aequation}
\end{gather}

Since $\Theta_{(n)} = 0$, this implies that $e^a$ is proportional to $\partial_Y$, and as $e^a$ is normalized, we have $e_a = \dfrac{1}{|| \ud Y||}\partial_a Y$. Imposing $e^a e_a = 1$ we obtain
\begin{gather}
    e^a = || \ud Y||(\partial_Y)^a  \, , \label{eq:Estatic}
    % =\left(1 - \frac{2M_{ms}}{r}\right)^{1/2} (\partial_r)^a\,,
\end{gather}
which gives us
\begin{gather}
    \mathcal{A} \Theta_{e} = -||\ud Y ||^2 \, \frac{2}{Y} \frac{ \partial_Y P}{\rho + P}\,.
\end{gather}

Therefore, replacing $|| \ud Y||^2$ by its meaning in terms of mass, $\Lie_n \Theta_{(n)} = 0$ corresponds to 
\begin{gather}
    \left(\epsilon - \dfrac{2\mu(Y)}{Y}\right)\frac{2}{Y} \frac{ \partial_Y P}{\rho + P} =  - 2 \frac{\mu(Y)}{Y^3}- 8 \pi P\,,
\end{gather}
or, alternatively,
\begin{gather}
  \frac{ \partial_Y P}{\rho + P}   = -\left(   \frac{\mu(Y)}{Y^2}+ {4\pi PY}\right)\left({\epsilon}- \frac{2\mu(Y)}{Y}\right)^{-1}\,, \label{eq:TOV}
\end{gather}
\Amv{which is what we call \emph{the unified TOV equation}. It reduces to the well-know TOV equation for spherically symmetric spacetimes when $\epsilon = 1$, and it corresponds to the  equation of hydrostatic equilibrium for planar and hyperbolic geometries, in the cases where $\epsilon =0$ and $\epsilon = - 1$, respectively.} This underlines the fact that the TOV equation is a hydrostatic equilibrium equation, and not an equation of state, as it is erroneously stated sometimes.

In order to determine $\mu(Y)$ we consider the $\Lie_e \Theta_{(e)}$ Raychaudhuri equation
\begin{eqnarray}
    \Lie_e \Theta_{(e)} &=&\mathcal{ B} \Theta_{(n)} - \frac{\Theta_{(e)}^2}{2} +\frac{1}{4} \left(\Theta_{(n)}^2 - \Theta_{(e)}^2 \right)\nonumber\\
    &+& \frac{\epsilon}{Y^2 } - 8 \pi T_{ab}n^a n^b\, , \label{eq:Lee}
\end{eqnarray}
which, by using $\Theta_{(n)} = 0 $, and Eq.~\eqref{eq:Estatic} lead us to
\begin{gather}
    ||\ud Y|| \partial_Y \left(\frac{2}{Y} ||\ud Y||\right) = -\frac{3}{Y^2}||\ud Y||^2 + \frac{\epsilon}{Y^2} - 8 \pi \rho\,.\label{eq:dYintermed}
\end{gather}
Substituting Eq.~\eqref{eq:muequation} into Eq.~\eqref{eq:dYintermed}, we obtain 
\begin{gather}
    \partial_Y \mu = 4 \pi \rho Y^2\,,\label{eq:massequation}
\end{gather}
which looks like the %known 
mass-energy equation of spherical symmetry.  {Here, it should be interpreted as the mass-energy equation in the spherical case, and as a mass-energy parameter equation in the planar and hyperbolic cases. Furthermore, Eqs.~\eqref{eq:massequation} and \eqref{eq:muequation} imply that if the \emph{weak energy condition (WEC) \cite{Martin-Moruno:2017exc} holds} only the spherically symmetric case admits static \Amv{ regular solutions. Indeed, as those solutions require $||\ud Y||^2 > 0$, %which 
that implies $\mu < 0$ for $\epsilon \leq 0$ and, as in regular spacetimes,
\begin{gather}
    \mu(Y) = 4 \pi \int_0^Y \rho(y) y^2 \ud y \,,  
\end{gather}
this imposes $\rho < 0$.}
}
 %This choice is convenient in order to treat all the three cases simultaneously and recover the mass function with the traditional factor.

With Eq.~\eqref{eq:massequation}, the last requirement to solve Eq.~(\ref{eq:TOV}) is the equation of state of the fluid, $f(\rho, P) = 0$ which should come from specific physical considerations.

\subsection{ Spacelike Killing Vector}\label{sec:subSKV}

In the spacelike Killing vector 
case, $\ud Y$ is timelike, the flow $n_a$ is orthogonal to the Killing vector, and the unitary base vector $e^a$ is parallel to it.\Amv{ This imposes no constraint on the sign of $\mu$ according to Eq.~\eqref{eq:muequation}, %therefore,
and thus there is no need to violate energy conditions in order to consider these solutions, thorougly studied in cosmology  \cite{Ellis:1998ct}.}

One dynamical equation is given by $\Lie_e \Theta_{(e)}=0$, which according to Eq.~\eqref{eq:Lee} gives:
\begin{equation}
   \mathcal{ B} \Theta_{(n)} - \frac{\Theta_{(e)}^2}{2} +\frac{1}{4} \left(\Theta_{(n)}^2 - \Theta_{(e)}^2 \right)
    + \frac{\epsilon}{Y^2 } - 8 \pi \rho = 0\,. 
\end{equation}
From Proposition~\ref{prop:killing}, we have $\Theta_{(e)} = 0$. Replacing  Eq.~\eqref{eq:B} {in Eq.~\eqref{eq:Lee}},  we obtain
\begin{gather}
\frac{3}{4}\Theta_{(n)}^2 - 3\sigma \Theta_{(n)} + \frac{\epsilon}{Y^2} = {8\pi \rho}\,,\label{eq:LieeThetae}
    \end{gather}
%or
and, {using Eq.~\eqref{eq:theta3}, we can express this Eq.~\eqref{eq:LieeThetae}} in terms of the volume expansion $\Theta_3$ obtaining
\begin{gather}
\frac{ \Theta_3^2}{3}- 3\sigma^2 = 8 \pi \rho - \frac{\epsilon}{Y^2}\,, \label{eq:evolution1}
\end{gather}
which corresponds \Amv{to the generalised Friedmann constraint equation for} the evolution of a homogeneous and anisotropic universe.  
%{In the case $\sigma=\epsilon=0$,  we may identify $\Theta_3 = 3H$ and we %obtain recover the usual Friedmann equation for a flat universe.}
\Amv{In the case $\sigma=0$,  we may identify $\Theta_3 = 3H$, and we %obtain 
recover the usual Friedmann equation for the flat ($\epsilon=0$) and open ($\epsilon=-1$) spatially isotropic universes. Notice though that $\sigma=0$ also yields anisotropic, cosmological solutions when the matter content is not a perfect fluid \cite{Mimoso:1993ym}}.

%%%%%%%%%%%%%%%%%%%%%%%%%%%%%
%We need another equation in order to close the system. Considering the equation 
%\begin{gather}
%    \Lie_e \Theta_{(n)} = \Lie_k \Theta_{(k)} - 
%    \Lie_l \Theta_{(l)} + \Lie_l \Theta_{(k)} - \Lie_k \Theta_{(l)}=0\,,
%\end{gather}
%we obtain
%\begin{gather}
%    \mathcal{B} \Theta_{(e)} - \frac{\Theta_{(n)} \Theta_{(e)}}{2} - 8 \pi %T_{ab}n^a e^b = 0\,,
%\end{gather}
%which fixes $T_{ab}n^a e^b = 0$ for our models.
%%%%%%%%%%%%%%%%%%%%%%%%%%%%%%%%%

The $\Lie_n \Theta_{(n)}$ Raychaudhuri equation gives the evolution of $\Theta_{(n)}$. Using Eq.~\eqref{eq:Lienthetan}
\begin{gather}
    \Lie_n \Theta_{(n)} = - \frac{\Theta_{(n)}^2}{2} - 
    \Theta_{(k)} \Theta_{(l)} - \frac{\epsilon}{Y^2}\, - 8\pi T_{ab}e^a e^b \Rightarrow \nonumber\\
     \Lie_n \Theta_{(n)} = -\frac{3}{4}\Theta_{(n)}^2 - \frac{\epsilon}{Y^2} - 8 \pi T_{ab}e^a e^b\,.\label{eq:LienThetan2}
\end{gather}

Subtracting Eq.~\eqref{eq:LieeThetae} from Eq.~\eqref{eq:LienThetan2}, we obtain
\begin{gather}
    \Lie_n \Theta_{(n)} = -3\sigma \Theta_{(n)} - 8 \pi \left(\rho + P\right)\,,\label{eq:evolution2}
\end{gather}
which, together with an equation of state relating $\rho$ and $P$ closes our system.
By adding half of the Eq.~\eqref{eq:evolution2} with one third of Eq.~\eqref{eq:evolution1}, we obtain:
\begin{eqnarray}
    \Lie_n \left(\frac{\Theta_3}{3}\right) &+& \left(\frac{\Theta_3}{3}\right)^2 =\nonumber\\ &-&2\sigma \left(\frac{\Theta_3}{3}+ \sigma \right)- \frac{\epsilon}{3Y^2} - \frac{4 \pi}{3} \left(\rho + 3P \right)\,. \label{Raychaudhuri_spacelike-K}
\end{eqnarray}

\Amv{Those homogeneous and anisotropic spacetimes belong to a subclass of Bianchi models~\cite{bianchi2001three,Ellis:1968vb}, with the case $\epsilon = 0$ corresponding
to Bianchi type I universes,   $\epsilon = -1$ corresponding to the Bianchi type III models, and $\epsilon =1$ to the Kantowski-Sachs spacetimes  \cite{Kantowski:1966te}.

An extensive classification and evolution analysis of this family of spacetimes is certainly worthy of interest, and a great deal of work has already been carried out in the literature in this connection \cite{Ellis:1998ct}.  \Mov{However, in }%In 
this work we are %though 
mainly focused on the hydrostatic equilibrium situations. %from now on.
% and will leave the cosmological characterization of the cosmological models for future work.
\Mov{Thus, }%So 
our interest will be directed to understanding whether it is possible to find a correspondence between the TOV equation of static equilibrium, and some condition applying to the spatially homogeneous models.  

Imposing staticity amounts in the present case to have  $\Theta_3=0$, $\Lie_n \Theta_3=0$, and $\sigma=0$ in Eqs. (\ref{eq:evolution1}) and (\ref{Raychaudhuri_spacelike-K}).  Reconciling the  reduced equations simply requires  
\begin{equation}
   \rho+P=0 \; ,
\end{equation} 
in the $\epsilon=\pm1$ cases, and has no realisation when $\epsilon=0$, as $\rho=0$ from Eq.(\ref{eq:evolution1}).  Hence we  conclude that the TOV condition interpreted as a cornerstone of stability yields the well-known equation of state characterising a cosmological constant in the non flat  cases. In hindsight, one could have \Mov{anticipated this result}%antecipated
, %but 
which emerges here in a self-consistent way. Moreover we see that strong energy condition (SEC) is violated for both cases $\epsilon=\pm 1$, whilst the  WEC is additionally violated for the $\epsilon=-1$ case, as %it 
follows  from Eq. (\ref{eq:evolution1}). 
} 

 {
\section{incompressible fluid solutions} \label{sec:PHS}

Using our unified TOV equation, Eq.~\eqref{eq:TOV}, we may look for static perfect fluid solutions for all three symmetries considered here. By choosing a timelike coordinate $T$ along the flow\Mov{, %\sout{by}
} making $n_a = -\alpha(Y)\ud T$\Mov{, } and  {the warp factor Y}%\sout{ \Mov{a $Y$ "areal scale"}}
, we obtain the following line element in the $(T,Y)$ coordinates:
\begin{gather}
    \ud s^2 = -\alpha^2(Y)\ud T + \frac{\ud Y^2}{\epsilon - \frac{2 \mu(Y)}{Y}}+ Y^2 \ud \Omega_\epsilon\,, \label{eq:dsTY}
\end{gather}
where $\ud \Omega_\epsilon = \left( \ud \theta^2 + S_{\epsilon}^2  \ud \phi^2 \right)$ and the functions $\alpha$ and $\mu$ will be given by solving Einstein equations, i.e., Eqs.~\eqref{eq:TOV} and \eqref{eq:muequation}.

Here, we will apply our unified treatment to find the analogs of Schwarzschild interior solution, that is, we will use the equation of state of an incompressible fluid $\rho = \rho_0$ constant. It is important to note that, as we have discussed in Sec.~\ref{sec:subTKV}, the static solutions with $\epsilon \neq 1$ violate the WEC, therefore we should take $\rho_0 < 0$ in those cases.

Equation~\eqref{eq:Aequation} implies
\begin{gather}
    \frac{\alpha'}{\alpha} = -\frac{P'}{\rho + P} \Rightarrow \alpha = \frac{c_0}{\rho_0 + P}\,,
\end{gather}
where $c_0$ is an integration constant that can be set by rescaling the time coordinate and the prime denotes $Y$ differentiation.

Equation~\Mov{%\sout{\eqref{eq:muequation}}
\eqref{eq:massequation}} gives us
\begin{gather}
    \mu(Y) = \frac{4 \pi \rho_0 Y^3}{3}\,,
\end{gather}
which we replace in Eq.~\eqref{eq:TOV} to obtain
\begin{gather}
    P(Y) = \rho_0 \left( \frac{2\sqrt{|\epsilon - \frac{Y_s}{Y_g}|}}{3\sqrt{|\epsilon - \frac{Y_s}{Y_g}|} - \sqrt{|\epsilon - \frac{Y_s Y^2}{Y_s^3}|}} - 1\right)\,.
\end{gather}
where $Y_g$ is the analog of the radius of the object and is the least positive number that satisfy $P(Y_g)=0$, $Y_s= \dfrac{8\pi \rho_0 Y_g^3}{3}$ is the analog of the Schwarzschild radius, although it can not be interpreted as a location since it will be a negative number. This gives
\begin{gather}
    \alpha = \frac{1}{2}\left(3\sqrt{\left|\epsilon - \frac{Y_s}{Y_g}\right|} - \sqrt{\left|\epsilon - \frac{Y_s Y^2}{Y_g^3}\right|} \right)
\end{gather}
which has a similar form to the interior Schwarzschild solution, where we only change the sign of the mass-energy parameter and change the value of $\epsilon$ in the formula. Of course the physical properties are very distinct, since %those
the solutions violate the WEC. 

In %the 
Figure~\ref{fig:ptov} we compare the pressure for the three cases. \Amv{From the slope of the curves, we notice that only the hyperbolic case presents $P'>0$, %which is 
%a consequence of 
compensating
the repulsive gravity force in this setup%, being compensated by a positive pressure gradient
. This is the opposite of the more familiar situations presented in the spherical and planar cases, where gravity is attractive, %and 
with $P'<0$ %in order to 
sustaining the weight of the configuration.} We can also see that the planar case admits a positive pressure for $0< Y < Y_g$. That means that, as long as mass-energy is negative, we may have static plane configurations over a finite $Y$ interval. On the other hand, the hyperbolic solution only admits positive pressure for $Y > Y_g$, \Mov{so }there is no %t an 
analog of the Schwarzschild interior solution for this foliation,\Amv{ although it can be interpreted as an exterior fluid solution to an internal void. It can thus be matched to a hyperbolic vacuum solution for $Y< Y_g$, %which can be 
found as a particular case in Ref. \cite{Maciel:2018tnc}:
\begin{gather}
 \ud s^2 = - \left(\frac{2m}{Y}-1\right)\ud t^2 +\frac{ \ud Y^2}{\frac{2m}{Y}-1}+Y^2(\ud\theta^2+\sinh^2 \theta \ud \phi^2)\,,
\end{gather}}
\noindent where the parameter $m = |\mu|$.
The peculiarities of the hyperbolic solutions with regard to the energy conditions are also found in \Mov{one of the coauthors' }%the 
work %from one of us 
\cite{Lobo:2009du}. 
\begin{figure}
\includegraphics[width=0.45\textwidth]{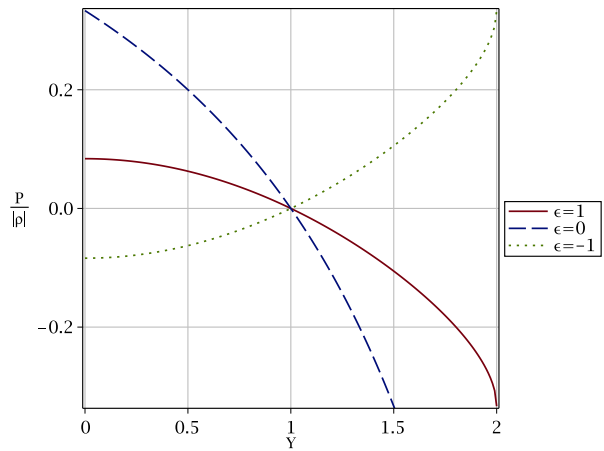}
\caption{Pressure as function of $Y$ for $Y_g =1$ and $|Y_s| = 0.25$ for $\epsilon = 1$ , $\epsilon = 0$  and $\epsilon= -1$.}\label{fig:ptov}
\end{figure}

The equation of state consisting of a negative energy density with a positive pressure might be achieved by some kind of phantom field, but a Lagrangian description of the fluid is beyond the scope of this work. \Amv{ However, %even 
our simple incompressible model, with constant energy density, but varying pressure, is reminiscent of a constant time surface of a McVittie or Shah-Vaidya spacetime, which admits a Lagrangian scalar field as source \cite{Abdalla:2013ara, C:2019zaw}. This suggests the possibility that there %are 
exist field models in the literature which can source the solutions presented in this paper.}

\Amv{We notice that the planar solutions may also represent a subclass of cylindrical solutions (see Refs.~\cite{Mena:2015zma, Bronnikov:2019clf}) if we select one coordinate along the plane to be periodic. Thus, our planar solution may be interpreted as a static cylinder of fluid with boundary given by $Y=Y_g$. At %From 
this surface, % on 
the solution must be matched with a vacuum solution.}

Actually, all fluid configurations found can be matched with the corresponding %\Mov{\sout{the}} 
static solutions presented in Ref. \cite{Maciel:2018tnc} which arise from applying Birkhoff theorem for external fluid sources that satisfy the  hypotheses of the theorem. In those cases there is matter content present outside, the most common examples being an electromagnetic field, and a cosmological constant. Therefore the matching surface will correspond to a surface where $P(Y)$ matches the pressure of the exterior solution, in a manner similar to the way in which  an uncompressible charged sphere is matched to a Reissner-Nordström solution in \cite{Arbanil:2014usa}.}

\section{Conclusion}\label{sec:Conclusion}

In this paper we analysed spacetimes with a two-dimensional maximally symmetric foliation sourced by a perfect fluid. We proved that in those cases, if there is a Killing vector orthogonal to the leaves, its two-expansion vanishes, which allows us to %take advantage of this fact in order to 
simplify our dynamical equations in terms of the two-expansion of a unitary vector orthogonal to the Killing field.

When the Killing vector $\chi^a$ is timelike, we find that the flow lines must be tangent to $\chi^a$, and as this is true at all times, the equations describe a hydrostatic equilibrium{, governed by a (generalised) TOV equation}. When the Killing vector is spacelike, %then 
we have instead a spatially homogeneous dynamical spacetime%, spacially homogeneous
. The result is a subclass of Bianchi universes, with only %two distinct shear eigenvalues
\Amv{one shear degree of freedom}. The corresponding equation gives the evolution of expansion and shear scalars.

We have discussed the geometric meaning of the mass-energy in such spacetimes, and our procedure matches the traditional mass parameter found in those cases by usual methods of integration of Einstein equations. Our 
approach relates the mass parameter to the geometrically defined quasi-local mass-energies 
of Misner-Sharp and Hawking-Hayward \Amv{by slightly changing its definition in order to apply it to our infinite mass-energy distributions.} \Mov{This %second
innovation is in itself a step towards addressing the open issue of defining mass/energy in gravitation and cosmology, c.f the recent  works of \cite{Faraoni2016}, and others \cite{Barzegar:2017ijz,Wang:2008jy}  on this subject.}

Using these concepts we could recover the physical interpretation of the geometrical quantities appearing in the equilibrium/evolution equations,  
translate the dual null formalism to the more usual relativistic fluid dynamics framework, and show that the TOV equation arises as a particular case of those equations. Henceforth the generalizations of the TOV equation appear automatically  by just setting $\epsilon = 0$ or $-1$ accordingly.  

\Amv{From this treatment it emerges the fact that} the only static fluid solutions that satisfy the WEC are the spherical ones, as the other two cases require a negative energy density. 

\Amv{In what regards the spatially homogeneous spacetimes, the hydrostatic equilibrium condition also implies a violation of SEC for the non planar solutions, constraining the equation of state for the perfect fluid to be that of a cosmological constant. }

%The same treatment also gives us the evolution equation of Bianchi spacetimes with two distinct shear eigenvectors \Amv{when the Killing vector is spacelike}. 

In order to illustrate the analogy between the planar, hyperbolic and spherical cases we studied the static solutions for an incompressible fluid. We found that, besides the known case of pherically symmetric spacetimes, we can {%\sout{find}
obtain} a static \Amv{interior} fluid configuration only in the case of planar symmetric spacetimes. \Amv{In the hyperbolic case, the static configuration is an exterior solution that can surround an inner vacuum region. 

Our unified way to describe three classes of spacetimes foliated with codimension-two leaves of constant curvature leads the way to further generalizations, as those spacetimes are of interest in many domains, from braneworld models to AdS/CFT duality. The adaptation of our formalism to  N-dimensional spacetimes is straightforward. 

\Mov{The introduction of the mass parameter to generalise Hawking-Hayward's mass also may have impact on other studies of compact objects in open backgrounds.} 
In addition, we discovered that some of the models require sources that violate energy conditions. Some popular models for modified gravity theories that aim to explain large scale phenomena\Mov{, such} as cosmological inflation and dark energy, also violate energy conditions~\cite{Capozziello:2013vna, Capozziello:2014bqa}, and there are several arguments in the literature suggesting those conditions should be abandoned as a criterion of viability  \cite{Barcelo:2002bv}. 
%For example, the authors of Ref.~\cite{Barcelo:2002bv}, state " it has become clear that there are quite reasonable looking classical systems, field theories that are compatible with all known experimental data, and that are in some sense very natural from a quantum field theory point of view, which violate all the energy conditions". 

One of the consequences of matter sources with such equations of state is that solutions with very different and intriguing properties arise, even considering the simplifying assumptions we made in order to obtain analytic results. \Mov{Seeking }%Looking for 
a Lagrangian model in the modified gravity literature that could generate our fluid solutions \Mov{appears as }%looks like 
an interesting continuation of our findings.
}
%%%%%%%%%%%%%%%%%%%%%%%%%%%%%%%%%%%%%%%%%%%%%%%%%%%%%%%%%%%%%%%%%%%%%%%%%%%%%%%%%%%%
\begin{acknowledgements}
 A.M. wish to thank the hospitality of Instituto de Astrofísica e Ciências do Espaço (IA), at the FCUL in Lisbon, where a part of this work was carried out. The work of M.Le~D. has been supported by Lanzhou University starting fund and by the Fundamental Research Funds for the
Central Universities (Grant No.lzujbky-2019-25). The work of JPM was supported by FCT/MCTES through national funds (PIDDAC) by this grant UID/FIS/04434/2019, and by the project PTDC/FIS-OUT/29048/2017. The authors wish to thank Xu Yumeng for helpful discussions.
\end{acknowledgements}
 \bibliography{shortnames,referencias}
\end{document}